\documentclass[a4paper]{article}

\usepackage[a4paper,top=3cm,bottom=2cm,left=2.5cm,right=2.5cm,marginparwidth=1.75cm]{geometry}
\usepackage[caption=false,font=footnotesize]{subfig}
\usepackage{float}
\usepackage{color}
\usepackage{cite}
\usepackage{amssymb,amsmath}
\usepackage{graphicx}
\usepackage[lined,ruled,commentsnumbered,linesnumbered]{algorithm2e}
\usepackage{color}
\usepackage{url}
\usepackage{placeins}
\usepackage{etoolbox}
\usepackage{comment}
\usepackage{blindtext}
\usepackage{multirow}
\usepackage{scrextend}
\addtokomafont{labelinglabel}{\sffamily}
\usepackage[official]{eurosym}

\usepackage{amsthm}
\newtheorem{definition}{Definition}
\newtheorem{theorem}{Theorem}

\begin{document}


\title{Secure and Privacy-Friendly Local Electricity Trading and Billing in Smart Grid}

\author{Aysajan~Abidin,~Abdelrahaman~Aly,~Sara~Cleemput,~and~Mustafa~A.~Mustafa
  \thanks{This work was supported in part by the Research Council KU Leuven: C16/15/058. In addition, this work was supported by the European Commission through KU Leuven BOF OT/13/070, H2020-DS-2014-653497 PANORAMIX and H2020-ICT-2014-644371 WITDOM.}           
 \thanks{A. Abidin, A. Aly, S. Cleemput and M.A. Mustafa are with the imec-COSIC research group, Dept. of Electrical Engineering (ESAT), KU Leuven, Belgium. e-mail: (\{aysajan.abidin, abdelrahaman.aly, sara.cleemput, mustafa.mustafa\}@esat.kuleuven.be).}
}
\date{\vspace{-5ex}}
\maketitle


\begin{abstract}
This paper proposes two decentralised, secure and privacy-friendly protocols for local electricity trading and billing, respectively. The trading protocol employs a bidding algorithm based upon secure multiparty computations and allows users to trade their excess electricity among themselves. The bid selection and calculation of the trading price are performed in a decentralised and oblivious manner. The billing protocol is based on a simple privacy-friendly aggregation technique that allows suppliers to compute their customers' monthly bills without learning their fine-grained electricity consumption data. We also implemented and tested the performance of the trading protocol with realistic data. Our results show that it can be performed for 2500 bids in less than five minutes in the on-line phase, showing its feasibility for a typical electricity trading period of 30 minutes.
\end{abstract}


\section{Introduction}
\label{Introduction}

The Smart Grid (SG) is an electricity grid supporting bidirectional communication between various entities and components in the grid. One of these components is a Smart Meter (SM) that allows real-time grid management~\cite{Farhangi2010}, thus improving the efficiency and reliability of the grid, and the seamless integration of various green energy sources such as Renewable Energy Sources (RESs) (e.g., solar panels) into the distribution grid. When these RESs generate more electricity than needed by their owners, the excess electricity is fed back to the grid. Currently, households get some compensation from their suppliers for such excess electricity at a regulated, low price. However, households with such excess electricity may be interested in selling directly to other consumers at a competitive price for monetary gains. Enabling that would also incentivise households to own RESs. 

To address this, a local electricity market that allows RES owners to trade their excess electricity with other households in their neighbourhood has been proposed in~\cite{Mustafa2016}. However, such a local electricity market has user privacy risks, since users' bids and offers reveal private information about their lifestyle~\cite{hart}. Although there are various proposals for an electricity trading market that allows users to trade with each other, none of them addresses the privacy concerns. The security and privacy concerns in such a local market have been analysed in~\cite{Mustafa2016}. However, no concrete solution has been proposed. In this work, we not only propose a concrete secure and privacy-friendly solution for such a market, but also implement and evaluate its performance using realistic data. This work extends our previous research~\cite{Abidin2016} in improving the bid-matching algorithm and providing a solution for privacy-friendly billing. 

Our specific contributions include: 

\begin{itemize} 

\item a decentralised protocol for local electricity trading that allows the market to identify the selected bids, calculate the clearance price, and compute the total amount of electricity traded by the users belonging to each individual supplier in an oblivious and secure manner; 

\item a concrete time-of-use privacy-friendly billing protocol that allows suppliers to calculate their users' monthly bills without accessing their fine-grained metering data; and

\item an implementation, evaluation and analysis of the proposed protocols using realistic metering data.

\end{itemize}
    
The rest of this paper is organised as follows: Section~\ref{Related Work} presents the related work. Section~\ref{Preliminaries} elaborates on the preliminaries. In Section~\ref{The Protocol}~and~\ref{Security Analysis}, we present our proposed protocols and analyse their security, respectively. Section~\ref{Computational Experimentation} gives the details on our implementation and simulation results. Finally, Section~\ref{Conclusions} concludes the paper.


\section{Related Work}
\label{Related Work}

Preserving users' privacy in SG has already been recognised as an important issue by the research community~\cite{McDaniel2009,Kalogridis:USaPP}. Below we provide a brief overview of state-of-the-art privacy-friendly protocols for electricity trading and billing.

\subsection{Privacy-friendly Protocols for Electricity Trading}
 
Various local electricity trading models have already been proposed~\cite{Lee2014, Tushar2016}. Bayram~et~al.~\cite{Bayram2014} gave an overview of such models, whereas Zhang~et~al.~\cite{ZHANG2017} summarised existing local electricity trading projects. Mengelkamp~et~al.~\cite{Mengelkamp2017:bidding_mechanisms} evaluated several market designs and bidding strategies to demonstrate that all the evaluated market scenarios offer economic advantages for the participating users. Mustafa~et~al.~\cite{Mustafa2016} performed a security analysis of such a local trading market and raised the security and privacy concerns associated with it.

There are already various solutions that partially address these concerns. Mengelkamp~et~al.~\cite{Mengelkamp2017:blockchain_market} designed a local electricity market on a private blockchain. They presented a proof-of-concept model including a simulation of a local blockchain-based energy market that allows users to bilaterally trade energy within their community. Kang~et~al.~\cite{Kang2017} proposed a similar trading mechanism among electric vehicles using a consortium blockchain technology combined with an iterative double auction mechanism designed to maximize social welfare. Aitzhan~and~Svetinovic~\cite{Aitzhan2017} implemented a decentralised electricity trading system using combination of blockchain technology, multi-signatures, and anonymous encrypted messaging as a proof-of-concept. Their system provides identity privacy of participating users and transaction security. Mihaylov~et~al.~\cite{NRGcoin} proposed a virtual currency, called NRGcoin, to convert locally produced energy directly to NRGcoins. In their proposed scheme, each local distribution system operator independently determines for each time slot the rates for energy consumption and production in the neighbourhood, based on the supply-demand balance at that current time slot. 

Rahman~et~al.~\cite{RAHMAN2017} proposed a secure bidding protocol for incentive-based demand response system. However, their protocol is not fully privacy-friendly as the bidding manager, acts like a trusted party and learns all users' bids. Kounelis~et~al.~\cite{Kounelis2017} introduced a platform named Helios that allows users to exchange energy in a decentralised manner using a blockchain technology and smart contracts. Uludag~et~al.~\cite{Uludag2015} proposed a distributed bidding system where only the winning bidder is disclosed to a service provider, whereas the bids of the other bidders are kept private. The same authors extended their solution to facilitate multi-winner auction mechanism~\cite{Balli2017}. Although these solutions provide security and verifiability as they are based on the blockchain technology, they do not fully address users' privacy concerns as transactions could be traced and linked to users.

One way to avoid this privacy leakage is to use Multiparty Computation (MPC)~\cite{Yao82}. Aly and Van~Vyve~~\cite{AV16} proposed an MPC-based auction mechanism that allows suppliers and generators to trade electricity on the day-ahead market in a secure and oblivious manner. Abidin~et~al.~\cite{Abidin2016} proposed an algorithm to match the supply and demand bids of a local electricity market in a privacy-friendly manner. However, they did not provide any solution for billing users.

\subsection{Privacy-friendly Protocols for Billing}

Petrlic~\cite{Petrlic2010} proposed each SM to have a trusted module so that the SM can act as a tamper-resistance device, and calculate its user's bills. Molina-Markham~et~al.~\cite{Molina-Markham2010} proposed to use zero-knowledge proofs to allow SMs to calculate their monthly bills without providing suppliers with metering data, while allowing them to verify the bills. Jawurek~et~al.~\cite{Billing_Petersen_Commitments} proposed each SM to have a plug-in component that generates signed commitments of metering data which are used by the supplier to verify the user's bill. A similar approach was also proposed by Rial~and~Danezis~\cite{Rial2011}. However, instead of having a plug-in component, the authors suggest to have a more powerful user device which could deal with more complex electricity tariffs. Danezis~et~al.~\cite{Danezis2011} improved the protocol in~\cite{Billing_Petersen_Commitments}. In their protocol, each SM adds some random noise to its final bill, so that the supplier can not deduct any useful information about the user's consumption pattern from the bill itself.

Data aggregation is another privacy-preserving technique. The existing work mainly utilises homomorphic cryptographic schemes~\cite{Li2010,Mustafa2015:MUSP,DEP2SA} or MPC~\cite{Mustafa2017} to aggregate metering data used for operational purposes. Bilogrevic~et~al.~\cite{bilogrevic2014} proposed a privacy-friendly aggregation of user profile data. Their protocol allows users to trade personal attributes obliviously. We employ a similar technique that allows the users to report the amount of electricity they traded at the local electricity market in a privacy-preserving manner, in order to  obtain the right amount of compensation from their contracted suppliers.

Unlike the aforementioned work, we propose a decentrilised, efficient and privacy-friendly protocol that allows local electricity trading among users. We extend the work presented in~\cite{Abidin2016} by including (i) a more detailed protocol description, (ii) a formalised security analysis, and (iii) a protocol for efficient and privacy-friendly user billing.


\section{Preliminaries}
\label{Preliminaries}


This section briefly describes the system model, the local electricity market, threat model, the assumptions and requirements introduced in~\cite{Mustafa2016} on which our protocols are based.

\subsection{System Model} 
\label{System Model}

A local electricity market includes the following entities (see Fig.~\ref{Local_electricity_market}). \textit{Renewable Energy Sources (RESs)} are local sources of electricity, e.g., solar panels, located on users' premises. The electricity generated by RESs is consumed by their owners, but surplus electricity is injected into the grid.
\textit{Smart Meters (SMs)} are advanced metering devices which measure the amount of electricity flowing from the grid to the house and vice versa. 
\textit{Users} consume electricity and are billed for this.
\textit{Suppliers} are responsible for supplying electricity to all users who could not get a sufficient amount of electricity from their own RES or on the local market. They are also obliged to buy the electricity which their customers did not trade on the market but still injected into the grid.
\textit{The trading platform} manages the trades. It consists of servers run by parties with competing interests.

We further classify the entities involved in our model as \textit{Dealers}, \textit{Computational parties (Evaluators)} and \textit{Output parties}. \textit{Dealers} are any subset of parties in charge of providing the input data and distributing them among the computational parties. In our protocols, the input data is provided directly by the bidders, i.e., users via their SMs. 
\textit{Evaluators} are the subset of parties performing the computations in a distributed manner. They operate on the data received from the dealers.  
\textit{Output Parties} are any subset of parties that have access to the output data -- the results of the computations performed by the evaluators. At the end of the computation stage, the evaluators send the output data  to the correct output parties. Note that some output data may be made public to all parties, while others only revealed to a specific subset,  e.g., the suppliers.
	
The selection and the number of evaluators depend solely on the application and it could be as many as the number of parties involved in the computation. This approach, however, would be costly in terms of performance. Thus, in our model we choose to have three computational parties. We assume that one party comes from the RES owners (bidders), one from the suppliers and a third one from a local authority, so the need for a trusted trading platform is eliminated, while the security and correctness of the trading operations are still guaranteed.


\begin{figure}[!t]
\centering
\includegraphics[trim= 0 5 0 10,clip=true,width=3.39in]{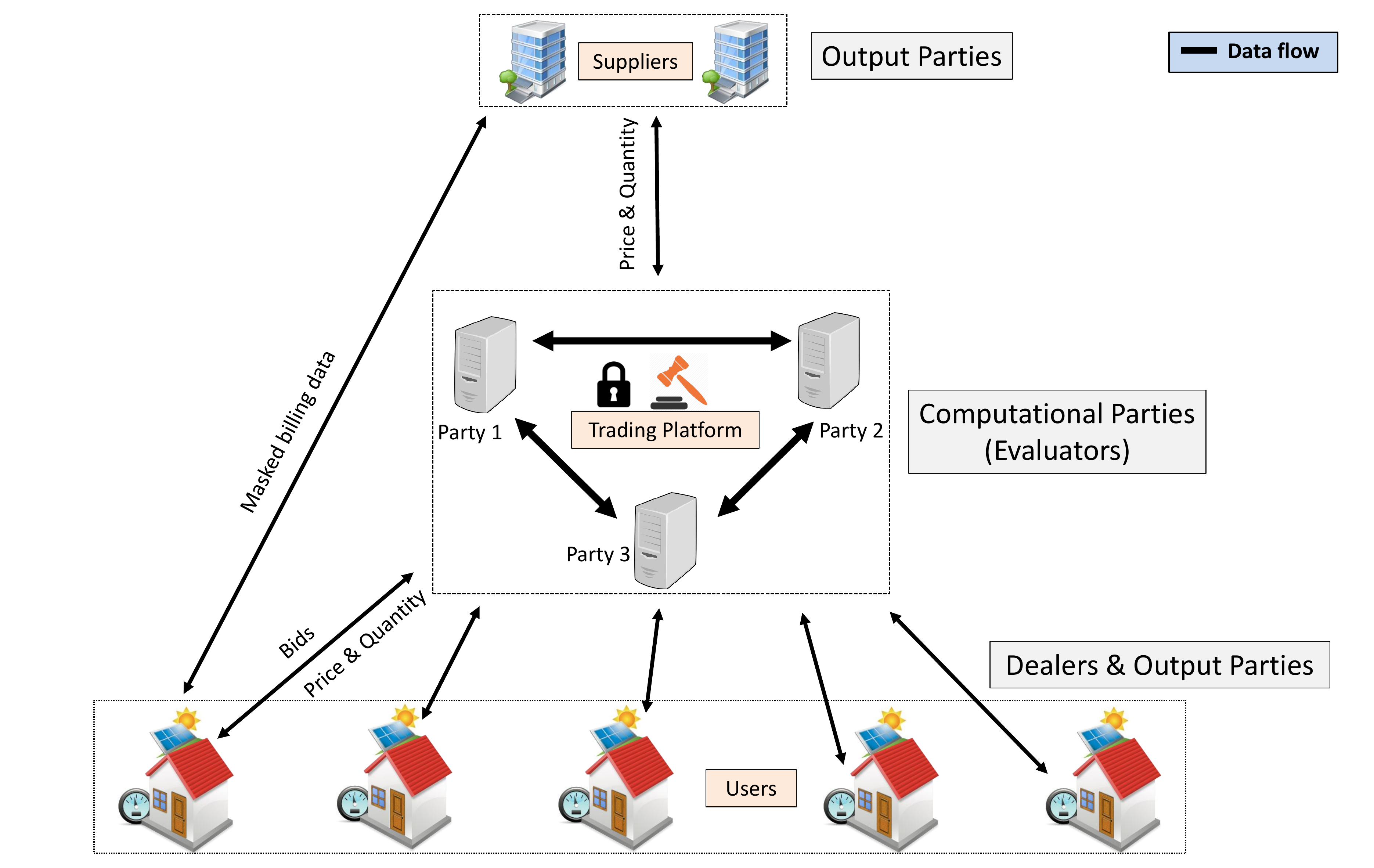}
\caption{A local electricity trading market.}
\label{Local_electricity_market}
\end{figure}



\subsection{Local Electricity Market Operation Overview}
\label{market_overview}

We briefly describe the local electricity market proposed by Mustafa~et~al.~\cite{Mustafa2016}. It consists of the following steps.

\begin{itemize}

\item[1.] \textit{Bid submission:} Prior to each trading period, users submit their bids to the trading platform. Each bid consists of the amount of electricity the user is willing to sell or buy during the trading period and for what price per unit.

\item[2.] \textit{Trading price \& winners selection:} The trading platform performs a double auction to determine the trading price, the amount of electricity traded, and the auction winners.
    
\item[3.] \textit{Informing users and suppliers:} The platform informs (i) the users about the amount of electricity they traded and the trading price, and (ii) the suppliers about the amount of electricity agreed to be traded by their respective users.
    
\item[4.] \textit{Delivering electricity:} During the electricity trading period the auction winners export (or import) the amount of electricity they traded on the local market. 
    
\item[5.] \textit{Calculating rewards and costs:} At the end of the trading period, each SM measures its user's imported and exported electricity for this period, which are used to calculate the users' rewards or costs. 
    
\item[6.] \textit{Settling payments:} Once the suppliers receive the imported and exported electricity values from their customers' SMs, they use them in conjunction with the users' trades for the trading period and the trading price to adjust the customers' bills in order to reflect the effect of the users participating in the local electricity trading market.
    
\end{itemize}


\subsection{Threat Model and Assumptions}
\label{Threat-Model-and-Assumptions}
The trading platform is an honest-but-curious entity. It follows the protocol specifications, but it may attempt to learn individual users' bids.
Users, suppliers and external entities are malicious. Users may try to modify data sent by their (or other users') SMs in an attempt to gain financial advantage, whereas suppliers may try to learn and modify users' bids in an attempt to influence the electricity trading price on the market. They may eavesdrop data in transit trying to learn confidential data and/or modify the data in an attempt to disrupt the SG.


We make the following assumptions: (i) each entity (e.g., SM) has a unique identity, (ii) SMs are tamper-evident, (iii) all entities are time synchronised, (iv) the communication channels between entities are secure and authentic, and (v) users are rational - they try to buy electricity for the cheapest price but sell excess electricity at the highest possible price.


\subsection{Privacy Requirements}


Our protocols should satisfy the following requirements: 
    
\begin{itemize}
    
\item \textit{Confidentiality of bids:} No entity (except the bid originator) should have access to the content of bids. 
    
\item \textit{Auction winner privacy:} The identity of a trading user (auction winner) should not be disclosed to any entity.
            
        
\item \textit{Minimum data disclosure:} Suppliers should access only the aggregate data traded by their consumers.

\item \textit{User privacy:} Suppliers should access only the monthly bill of all their customers (including the trading users).
    
\end{itemize}

\subsection{Multiparty Computation (MPC)} 
\label{MPC}

Secure MPC allows any set of mutually distrustful parties to compute any function such that no party learns more than their original input and what can be inferred from the output. In short, parties $P_{1},...,P_{n}$ want to compute $y= f(x_{1},...,x_{n})$, where $x_{i}$ corresponds to the secret input of party $P_{i}$, in a distributed fashion with guaranteed correctness such that $P_{i}$ learns only $y$ and what can be inferred from $y$. Furthermore, by using arithmetic circuits, any functionality can be constructed on MPC. Notice that any oblivious functionality built in this way would be as secure as the underlying MPC protocols used for its execution. For our protocol design we make use of the following well-known functionality construction: \textit{secure comparison}~\cite{CH10}, \textit{secure sorting}~\cite{HKICT12} and \textit{secure permutation}~\cite{CKKL99}.

\begin{table}[t]
\caption{Notation.}
\label{table Notations}
\centering
\scalebox{1.0}{
\begin{tabular}{l @{\hskip 0.10in}  p{0.8\linewidth}}
\hline
Symbol & Meaning \\
\hline

~~~~$ t_{i} $ & $i$-th time slot \\ 

~~~~$ u_j $ & $j$-th user \\

~~~~$ x_{j,i} $ & $j$-th user's bill for $ t_{i} $ \\

~~~~$ [q]_{j} $ & electricity volume in absolute terms from the $j$-th bid\\ 

~~~~$ [p]_ {j}$ & unit price enclosed in the $j$-th bid \\

~~~~$ [d]_{j} $ & binary value corresponding to the $j$-th bid: $1$ indicates a demand bid and $0$ a supply bid \\

~~~~$ [s]_{j} $ & unique supplier identifier $s \in \{1,..,|S|\}$ where $S$ is the set of all suppliers. Moreover, $s$ is encoded in a $\{0,1\}$ vector, i.e, $[s]_{jk} \leftarrow 1$ on the $k$-th position corresponds to the suppliers unique identifier, and $[s]_{jk} \leftarrow 0$ otherwise, for all $j \in B$. \\ 

~~~~$ [b]_{j}$ & bid's unique identifier from the $j$-th bid \\ 

~~~~$ [\phi] $ & volume of electricity traded on the market for $ t_{i} $ \\

~~~~$ [\sigma] $ & trading price (price of the lowest supply bid) for $ t_{i} $\\

~~~~$ [a]_{i}$ & binary value: $1$ $\leftarrow$ bid $i$ accepted, $0$ $\leftarrow$ bid $i$ rejected \\

~~~~$ [S]^{\phi} $ & set of the volume of electricity traded by supplier affiliation where $[s]^{\phi}_{i}$ stands for the summation of all the accepted bids from users affiliated to the supplier $i$, for all $i \in S$\\

\hline
\end{tabular}
}
\end{table}

\subsection{Notations}

We use square brackets to denote either encrypted or secretly shared values, e.g., $[x]$. Assignments that are a result of any securely implemented operation are represented by the use of the infix operator, e.g., $[z] \leftarrow [x]+[y]$. This extends to any operation over securely distributed data, as its result would be of a secret nature too. Vectors are denoted by capital letters. For a vector, say $B$, $B_{j}$ represents its $j$-th element and $|B|$ its size. 
The bids originated by SMs are considered as input data. Each bid is a tuple $([q],[p],[d],[s],[b])$ and $B$ is the vector of all bids. Table~\ref{table Notations} lists the notations used in the paper.

We assume all bid elements belong to $\mathbb{Z}_{M}$, where M is a sufficiently large prime so that no overflow occurs, and the number of bids or at least an upper bound on them is publicly known. Any other data related to the bid is kept secret. Note that in case the protocol admits a single supply and a single demand bid per SM, the computation of this upper bound is trivial. Markets could opt for enforcing all participants' SMs to submit a bid regardless of whether they participate or not at the $t_{i}$ market clearance. Let $\top$ be a sufficiently big number such that it is greater than any input value from the users, but $\top << M$. In this scenario, non-participating SMs would have to replace their input values by $[0]$ and $[\top]$ accordingly.


\section{Privacy-friendly Protocols for Local Electricity Trading and User Billing}\label{The Protocol}

In this section, we propose two novel protocols for privacy-friendly electricity trading and billing, respectively.

\subsection{Electricity Trading Protocol}
\label{sec:trading-protocol-description}

Our trading protocol covers the first three steps of the local electricity market summarised in~Section~\ref{market_overview}, namely \textit{Bid submission}, \textit{Trading price \& winners selection} and \textit{Informing users and suppliers}. Prior to the description of these steps, we first introduce the time-frame of the trading. Private bids for trading period $t_{i}$ have to be submitted before the beginning of $t_{i-2}$. The auction for $t_{i}$ is computed at period $t_{i-2}$ and the outcome is announced to the users and suppliers before the end of $t_{i-2}$. This is done to allow suppliers to trade in the wholesale market during $t_{i-1}$ so that they can adjust their wholesale deals according to the local market outcome.

\subsubsection{Bid submission (for $t_i$)}

It takes place before the start of $t_{i-2}$ and consists of two tasks: bidding and randomisation.

\begin{itemize}

\item \textbf{\texttt{Bidding:}} Each dealer, i.e., SM, prepares its individual bid and sends (parts of) it to the computational parties, i.e., the trading platform. The preparation takes into account various factors such as available (or required) extra electricity, market conditions, as well as the type of underlying cryptographic primitives used in the protocol. In the case of using a secret sharing scheme as the primitive, each dealer generates shares of its bid using a linear secure secret sharing scheme of choice, e.g.,the scheme in~\cite{Shamir79}. The number of generated shares is equal to the number of the computational parties, three in our case. This is the only input required from the dealers. Then, they send the corresponding shares of their bids to the respective computational parties for evaluation.

\item \textbf{\texttt{Randomisation:}} To randomly permute the dealers' input, upon reception, each computational party multiplies each share (ciphertext) with a column of a randomised permutation matrix. These matrices can be precomputed in an ``offline phase''. This can be achieved by the intervention of a trusted dealer that is not directly involved at any level of the computations~\cite{DPSZ12}. The amount and the purpose of the randomly generated numbers depend on the underlying security model and primitives used by the market. Since the permutation matrix generation used by Bogdanov~et~al.~\cite{BLW08} can be executed ``offline'', we use this technique in our protocol.

\end{itemize}

\subsubsection{Trading price \& winners selection (for $t_i$)}
It takes place at period $t_{i-2}$ and consists of a single task: the market clearance.

\begin{itemize}

\item \textbf{\texttt{Market Clearance:}} The computational parties calculate the trading price and the traded electricity volume, and identify the accepted and rejected bids in a data-oblivious fashion. Algorithm~\ref{algo:auction} gives a detailed overview of our privacy-friendly market clearance. The computational parties calculates the trading price $[\sigma]$, the volume of electricity traded $[\phi]$ and the vector of adjudicated demand and supply bids $[A]$. They do it by obliviously calculating the aggregation of the demand bids $[\delta]$, and then iterating over the set of all bids in $B$ using their volume to match $[\delta]$. To access the vector of accepted supply bids, it is enough to compute $[A]_{j} \times (1-[d]_j) \times [b]_{j}$. To find the vector of accepted demand bids, it is sufficient to calculate $(1-[A]_{j}) \times ([d]_j) \times [b]_{j}$.

\end{itemize}


\begin{algorithm}[t]
\scriptsize
\label{algo:auction}
 
  \SetAlgorithmName{Algorithm}{y}{x} \KwIn{Vector of $n$ bid tuples $B=([q],[p],[d],[s],[b])$}
  \KwOut{Trading price $[\sigma]$, volume of traded electricity  [$\phi$], vector of accepted bids $[A]$ of size $|B|$, vector of aggregated volume traded by supplier $S^{\phi}$ of size $|S|$} 
 
	\For{$j \leftarrow 1$ \KwTo $n$}{
		$[\delta] \leftarrow [\delta] + [q]_{j} \times [d]_{j};$ \\
	}
	$[\nu] \leftarrow [0];$\\
	$[S^{\phi}] \leftarrow \{0_{1},...,0_{|S|}\};$\\
	$[A] \leftarrow \{0_{1},...,0_{|B|}\};$ \\
	\For{$k \leftarrow 1$ \KwTo $n$}{
		$[c] \leftarrow [\nu]<[\delta];$\\		
		$[\sigma] \leftarrow ((1-[d]_{j}) \times [c])\times ([p]_{j} -[\sigma]) +[\sigma];$ \\			
		$[\phi] \leftarrow ((1-[d]_{j}) \times [c]) \times[q]_{j} + [\phi];$\\	
		\For{$k \leftarrow 1$ \KwTo $|S|$}{
			$[s]^{\phi}_{k} \leftarrow ([s]_{jk} \times ((1-[d]_{j})\times[c]) \times [q]_{j} + [s]^{\phi}_{k};$\\
			
		}
		$[a]_{j} \leftarrow [c];$\\

		$[\nu] \leftarrow [\nu] + [c]\times[q]_{j};$\\	
		
	}

 \caption{Market Clearance}
\end{algorithm}


\subsubsection{Informing users and suppliers (for $t_i$)}

It takes place before the start of $t_{i-1}$ and consists of two tasks: informing users and suppliers, respectively.

\begin{itemize}

\item \textbf{\texttt{Informing users:}} To hide the order of the bids, the computational parties shuffle again the vector of all bids $[B]$, together with the associated vector $[A]$. Then, they use the \texttt{open} operation of the underlying MPC primitive on the trading price (for $t_{i}$) $[\sigma]$ and on all the bids $[b]_{j}$, for all $j \in B$. The computational parties then send the shares corresponding to the tuple $B_{b_{j}}$ to the dealer (SM/user) that originated the bid identified by $b_{j}$. The SM then reconstructs the shares and learns if the bid was accepted or rejected, as well as the trading price $\sigma$.

\item \textbf{\texttt{Informing suppliers:}} The computational parties send the shares of the volume aggregation $S^{\phi}_{j}$, for all $j \in S$, to the corresponding suppliers which reconstruct these shares to learn their corresponding aggregate volumes of electricity traded by their customers for $t_i$. The suppliers also learn the market trading price $\sigma$.

\end{itemize}

\textit{Correctness:} The general goal of the protocol is to compute the trading price and to identify the accepted and rejected bids for each trading period. Any supply bid below the trading price, and any demand bid above this price is automatically accepted. The market equilibrium is identified when the price of a given supply allocation surpasses the price of the next cheapest available demand allocation.  

In our protocol (see Algorithm~\ref{algo:auction}), we sort all the bids regardless of whether they are demand or supply bids. We then identify and select bids until the aggregate demand ($[\delta] \leftarrow \sum^{|B|}_{j} [q]_{j}\times [d]_{j}$) is matched (note that to maintain secrecy we iterate over the set of all bids) choosing the bids in ascending order of price. If a supply bid is selected, this implies that there is no supply bid that could be allocated to reduce $[\delta]$, and hence is not part of the market clearance. Using $[d]_{j}$ cancels the supply bid's effect over $[\delta]$, and  provides us with sufficient tools to identify it. The opposite occurs when a demand bid is selected. The bids that were used to reduce $[\delta]$ can be identified, and correspond to all the supply and demand bids with prices below and above the trading price, respectively. From this, the set of accepted bids follows. The trading price is set to the price of the last selected supply bid. 

\textit{Complexity:} The complexity of the protocol grows linearly with the number of bids. The complexity of Algorithm~\ref{algo:auction} is $\mathcal{O}(|B| \times |S|)$. Note that the number of suppliers rarely varies over time, and is of a limited size. As an additional note, secure vector permutation can be achieved in $\mathcal{O}(n \times log(n))$, where $n$ is the size of the vector. Moreover, the sorting method used by our protocol can achieve $\mathcal{O}(n \times log(n))$.


\subsection{Privacy-Friendly Billing Protocol}
\label{Privacy-Preserving Billing}

As mentioned before, the suppliers use the imported and exported electricity values from their customers' SMs, in conjunction with the users' trades for the trading period and the trading price, to adjust the customers' bills. However, providing such information to the suppliers poses privacy threats. Therefore, we propose the following private reporting mechanism that both facilitates the calculation of the correct bill and preserves the users' privacy. 

Let $L\in\mathbb{N}$ be the total number of trading a user has done at the local market during a billing period (i.e., there are $L$ trading periods at the local market during one billing period). Then each user $u_j$ has a vector $X_j=(x_{j,1},\cdots,x_{j,L})$ of $L$ values which correspond to the user's bill at each period $t_i,\,i=1,\cdots,L$. These values can be calculated by the SM locally, using the total amount of electricity measured at each period and the amount of electricity traded at that period. The goal is to report $x_{j,i}$s in such a manner that it preserves the privacy of the values but still allows the supplier to calculate the user's monthly bill, $\texttt{Bill}_j=\sum_{i=1}^Lx_{j,i}$. 

Let $M$ and $\mathbb{Z}_M$ be as before. Then, in each billing period, each user $u_j$ randomly selects $L$ elements  $s_{j,i}\in\mathbb{Z}_M$, for $i=1,\cdots,L$ such that $\sum_{i=1}^Ms_{j,i}\equiv 0 \mod M$.  At the end of each trading period, $u_j$ masks $x_{j,i}$ as $c_{j,i}\equiv x_{j,i}+s_{j,i}\mod M$. 
The user then sends $c_{j,i}$ to the supplier. Upon receiving all $c_{j,i}, \,i=1,\cdots,L$, from $u_j$, the supplier computes the monthly bill for the user $u_j$ as $$\sum_{i=1}^Mc_{j,i}\equiv \sum_{i=1}^M(x_{j,i}+s_{j,i})\equiv \sum_{i=1}^Mx_{j,i}\mod M.$$  

The random elements $s_{j,i}$ elements function as one-time masks so that $c_{j,i}$s do not reveal information about $x_{j,i}$s, for $i=1,\cdots,L$.


\section{Security Analysis}
\label{Security Analysis}
We analyse the security of our trading protocol (cf. Section~\ref{sec:trading-protocol-description}) in the Universal Composability framework~\cite{Canetti2000}. In this framework, the ideal functionality of MPC is modeled as Arithmetic Black Box (ABB)~\cite{DN03}. 
\begin{definition}[ABB Functionality $\mathcal{F}_{\textit{ABB}}$]
\label{def:FABB}
The ideal functionality $\mathcal{F}_\textit{ABB}$ for MPC is defined as follows:
\begin{itemize}
\item $\mathtt{Input}$: Receive a value $\alpha \in \mathbb{Z}_M$ and store $\alpha$. 
\item $\mathtt{Share}(\alpha)$: Create a share $[\alpha]$ of $\alpha$.
\item $\mathtt{Product}([\alpha], [\beta])$: Compute $\gamma = \alpha \times \beta$ and store $[\gamma]$.
\item $\mathtt{Compare}([\alpha],[\beta])$: Compare $\alpha$ and $\beta$, and return $0$ if $\alpha<\beta$ and 1 otherwise.
\item $\mathtt{Equal}([\alpha],[\beta])$: Check if $\alpha=\beta$; return $1$ if $\alpha=\beta$, 0 otherwise.
\item $\mathtt{sRand}()$: Sample $r \xleftarrow{R}\mathbb{Z}_M$ and store $[r]$. 
\item $\mathtt{Permute([X])}$: Given an input $[X]\in\mathbb{Z}_M^n$ return a random permutation of it. 
\item $\mathtt{Open}([\alpha])$: Send the value $\alpha$ to all players.
\end{itemize}
Addition and scalar multiplication are denoted by  their corresponding conventional symbols $+$ and $\times$.
\end{definition}
ABB can be thought of as a generic procedure for secure distributed computation, and provides us with abstraction of the details of MPC operations and of secret sharing. Any number of players can send their private input to ABB to compute any computable function on their private inputs. The computation results are stored in the internal state of ABB to be used in the subsequent computations. Stored values are only made public if enough number of players agree to it. 

\begin{definition}[UC-security \cite{C00}]
\label{def:UC-security}
 A real protocol $\pi$ is UC-secure if, for any adversary $\mathcal{A}$, there exists a simulator $\mathcal{S}$ for which no environment $\mathcal{Z}$ can distinguish with a non-negligible probability if it is interacting with $\mathcal{A}$ and $\pi$ or $\textsf{S}$ and the ideal functionality $\mathcal{F}$.
\end{definition}

\begin{definition}[Universal Composition \cite{C00}]
\label{UC-theorem}
Let $\pi$ and $\rho$ be two protocols such that $\rho$
$\epsilon_1$-UC-emulates $\mathcal{G}$ and $\pi\circ\mathcal{G}$ $\epsilon_2$-UC-emulates
$\mathcal{F}$, when using $\mathcal{G}$ as a subroutine. Then
$\pi\circ\rho$ $(\epsilon_1+\epsilon_2)$-UC-emulates $\mathcal{F}$, when using
$\rho$ as a subroutine.
\end{definition}

\begin{definition}[Ideal Functionality of the Local Electricity Market $\mathcal{F}_\textit{LEM}$]
Given a set of private input bids, the ideal functionality $\mathcal{F}_{\textit{LEM}}$ computes the trading price, the winners, and the aggregate data of winners per supplier, without learning anything but the trading price. 
\end{definition}

\begin{theorem}
Let $\pi$ be the local electricity trading protocol presented in Section \ref{sec:trading-protocol-description}. Then $\pi$ securely emulates $\mathcal{F}_{\textit{LEM}}$.
\end{theorem}
\begin{proof} 
The proof is a straightforward application of the Universal Composition Theorem~\ref{UC-theorem}, since (i) each ABB operation UC-emulates its corresponding ideal functionality and (ii) the local electricity trading protocol $\pi$ consists only of the ABB operations that are securely composable.  
\end{proof}




\section{Experimentation Results}
\label{Computational Experimentation}

In this section we provide the details of the implementation of our local electricity trading protocol and the results of its evaluation under realistic scenario configurations. Next we elaborate on the data generation process, the experimentation settings and the results.

\subsection{Data Generation} 
\label{Data Generation}

For our experiments, we used realistic data from Belgium. First we picked a time slot and date, i.e., 13:00h-13:30h on 5-th of May 2016, during which 2382~MW solar electricity was generated in Belgium by solar panels with total capacity 2953~MW~\cite{Elia_PV_electricity} -- on average each solar panel has produced electricity approximately equal to 81.66\% of its capacity. The average electricity consumption data of a Belgian household during the same time slot was 0.637~kW~\cite{Av_Cons_Data_BE}. Thus, for each user we generated a random consumption data for this time slot with mean equal to the average consumption data, i.e., 0.637, standard deviation equal to 0.20 and variance equal to 0.04. Then, we randomly chose 30\% of the users to have installed a solar panel at their homes, and each of the solar panels is randomly assigned one of the following electricity generation capacities: 2.3, 3.6 or 4,7 kW. After that, we randomly generated the electricity output of each solar panel during this time slot with a mean equal to the capacity of the given solar panel multiplied with the efficiency factor for the time slot, i.e., 81.66\%, standard deviation equal to 0.20 and variance equal to 0.04. Once we generated the electricity consumption and generation data for each user with a solar panel, we simply subtracted the latter from the first value to find the amount of excess electricity each user has. 

We assumed that there are 10 different suppliers available in the market and randomly assigned a supplier to each user. The retail electricity sell price of the suppliers is set to $0.20$~\euro/kWh and the retail buy price is set to $0.04$~\euro/kWh. For the bid price selection we used the following rational steps. We divided the retail electricity sell and buy price difference into nine ranges each including several (overlapping) prices, e.g., range 2 includes three prices: $0.04$, $0.05$ and $0.06$~\euro/kWh, whereas range 7 includes four prices: $0.17$, $0.18$, $0.19$ and $0.20$~\euro/kWh. Then, for each user, depending on how much excess electricity he/she has for sell or he/she wants to buy, we picked randomly one of the prices from the appropriate price range. For selecting the appropriate price range we assumed that the users are rational, i.e., if they have a lot of excess electricity to sell, they would choose a lower asking price, so they could sell it all. In contrast, if they have a little excess electricity to sell, they would ask for a higher price since selling it for a cheap price will not allow them to make a high profit by selling it at the market compared to selling it directly to the supplier.  

In summary, for each user we generated the following data items: unique user ID, amount of electricity for the bid, bid price, indicator if the bid is a supply or demand bid, and ID of the user's contracted supplier.


\subsection{Characteristics, Settings and Environment}

We executed our experimentation using the MPC Toolkit proposed in~\cite{Aly15} which is based on BGW~\cite{BGW88}. This library includes all the underlying crypto primitives and other sub-protocols we report on, together with our own introduced code. The library was compiled with NTL (Number Theory Library)~\cite{Shoup01} that itself was compiled using GMP (GNU Multiple Precision Library). These two libraries are used for the modulo arithmetic that is extensively used by the underlying MPC protocols. Each instance of the prototype comprises two CPU threads. One manages message exchanges exclusively and the other executes the protocol and the related cryptographic tasks. The application itself is not memory demanding, with each instance requiring approximately $1$ MB of allocated memory at any time, during our most memory demanding test which included 2500 bids.

Our prototype was built in C++ following an object oriented approach, with modularity and composability in mind. It has an engine that separates communication and cryptographic tasks. Table~\ref{tab:primitives} shows the detailed list of the primitives used in our implementation. All our tests were performed under a 3-party setting, with two available cores for each instance. We ran our tests starting with a baseline of a realistic scenario with 100 bids and then monotonically increasing the number of bids until they reached 2500 bids. Each test scenario was repeated $10$ times to reduce the impact of the noise. The values reported in the rest of this section correspond to those of the resulting averages. Finally, in all of our tests, bids have been randomly distributed between $10$ different suppliers. We executed our tests on a single 64-bit Linux server with 2*2*10-cores with Intel Xeon E5-2687W microprocessors at 3.1GHz and 25 MB of cache, and with memory of 256 GB.


\begin{table}[t]
  \centering  \caption{List of primitives used in our bidding protocol prototype}
  \scalebox{1.0}{
  \begin{tabular}{l l}
  \hline
   \textbf{Primitive}       &   \textbf{Protocol} \\
  \hline
	Sharing                 &   Shamir Secret Sharing~\cite{Shamir79} \\

    Multiplication          &   Gennaro et al.~\cite{GRR90}   \\

    Inequality Test         &   Catrina and Hoogh~\cite{CH10}   \\ 

    Random Bit Generation~  &   Damg{\aa}rd et al.~\cite{DKNT06}   \\ 

    Sorting: QuickSort ~  &   Hamada et al.~\cite{HKICT12}  \\ 

    Permutation: Sorting Network ~  &   Lai et al.~\cite{LWZ11}  \\ 
 \hline
  \end{tabular}}
   \vspace{-.2cm}
\label{tab:primitives}
\end{table}


\subsection{Results} 
\label{Numerical Results}

Our trading protocol prototype requires bit randomization for the comparison methods. The task of generating such values could be executed beforehand, in an ``offline'' phase. The ``online'' phase would execute the remaining tasks and would utilise the randomization values generated during the ``offline'' phase. Figure~\ref{test_result_2} shows such breakdown for all the test scenarios considered. We measured the computational cost at every test instance. For the 2500 bids case, our prototype took 678.50 seconds for either sending or waiting to receive other parties' messages (note that our prototype is synchronous) and 215.52 seconds for the remaining tasks, e.g., crypto primitives. In other words, approximately $ 75\%$ of the computational time was dedicated to transmission related tasks. Moreover, the asymptotic behaviour on the growth of the computational time seems to follow the behaviour included in the theoretical complexity analysis presented in Section~\ref{sec:trading-protocol-description}

Table~\ref{tab:execution} shows a more complete breakdown of our results. In short, the 2500 bids instance total time on the online phase is less than 5 minutes, while when the offline phase is considered, it is less than 15 minutes. In both cases this is less than the typical 30 minutes duration of a trading period $t$. Note that during our tests, approximately $95\%$ of the computational time was spent on sorting the bids. Although the number of suppliers rarely changes with time and it is relatively low, in our tests we used a fixed value of $10$. Given that they are involved in the other stages of the protocol, their influence on the number of suppliers is quite limited. This means that our algorithm can be easily adjusted to other realistic scenarios without much overhead, for bigger suppliers settings.

As for the simple billing protocols, it is computationally cheap and does not incur any communication overhead. The only requirement is that in each billing period, each user (or SM) needs to choose $L$ random elements from $\mathbb{Z}_M$ such that they add up to $0\mod M$.


 \begin{table}[t]
  \centering
  \caption{Overall Results}
  \begin{tabular}{ l c c c c c c|}
    \hline
    \textbf{Bids}   &   \textbf{Number of}~          & \textbf{Number of}  & \textbf{CPU time}    & \textbf{Online phase} \\
    ~   &   \textbf{com. rounds}~          & \textbf{comparisons}  & \textbf{(in sec)}    & \textbf{(in sec)} \\
    \hline
	    100             &   ~~$\approx 1.40\cdot 10^5$      & $965$                     & ~~$2.96$                      & ~~$1.01$ \\ 
        500             &   ~~$\approx 1.96 \cdot 10^6$     & $14628$                   & ~$40.40$                      & ~$11.35$ \\ 
	    1000            &   ~~$\approx 7.03 \cdot 10^6$     & $53508$                   & $147.76$                      & ~$39.80$ \\ 
        1500            &   $\approx 15.61 \cdot 10^6$      & $118956$                  & $320.79$                      & ~$86.14$ \\
        2000            &   $\approx 26.97 \cdot 10^6$      & $208132$                  & $562.50$                      & $145.78$ \\
        2500            &   $\approx 43.15 \cdot 10^6$      & $330912$                  & $894.01$                      & $235.82$ \\
 \hline
  \end{tabular}

\label{tab:execution}
\end{table}


\begin{figure}[t]
\centering
\includegraphics[trim= 0 0 0 0,clip=true,width=3.39in]{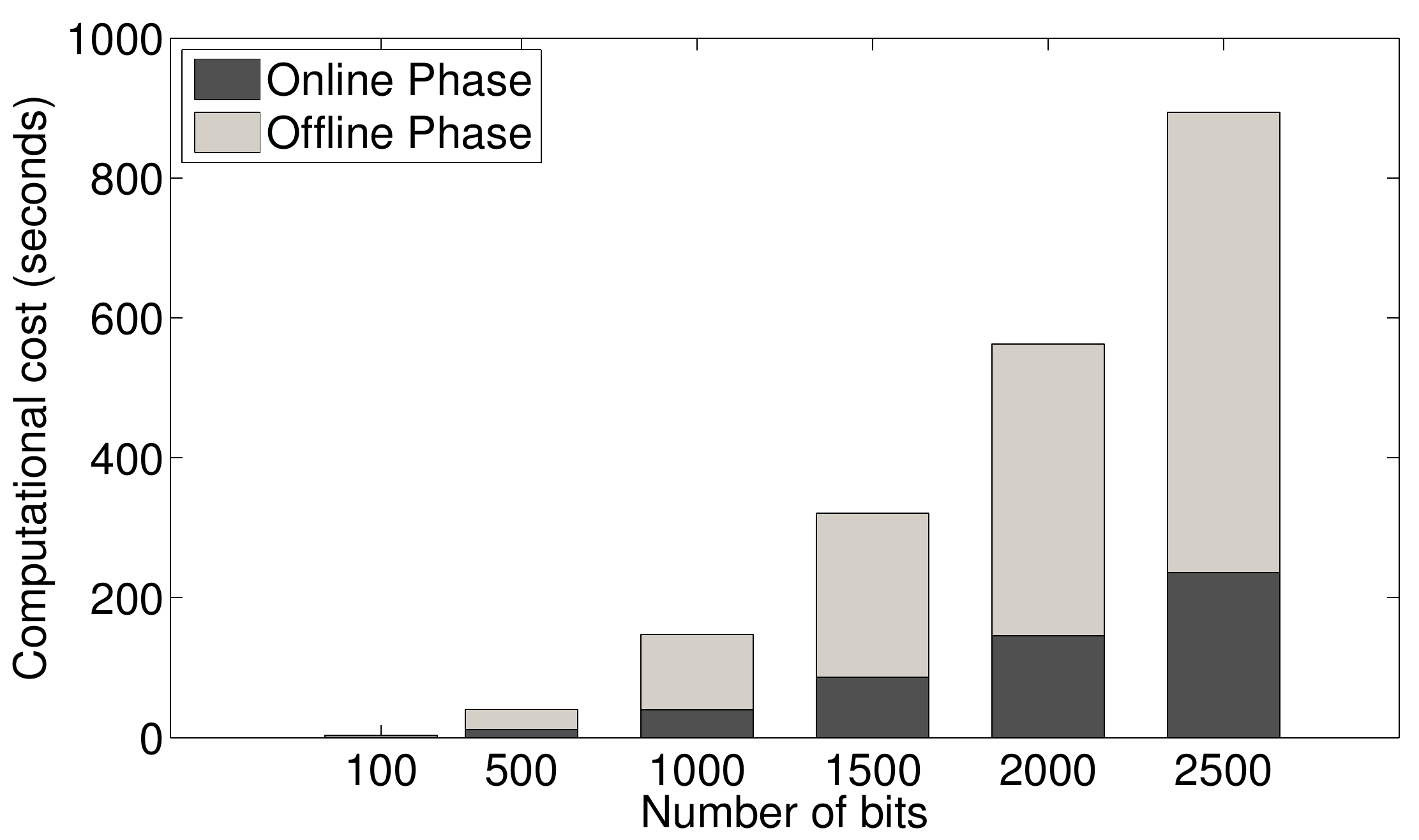}
\caption{Numerical results for our trading protocol.}
\label{test_result_2}
\end{figure}




\section{Conclusions}
\label{Conclusions}

In this paper we proposed two privacy-friendly protocols for local electricity trading and billing, respectively, that allow users to trade their excess electricity among themselves as well as calculate their monthly bill. Our first protocol employs a bidding scheme based upon MPC, and the selection of the bids and the calculation of the electricity trading price are performed in a decentralised and oblivious manner. Our second protocol uses private data aggregation technique to calculate the users' monthly bills locally and allows the suppliers to compute only the final monthly bill per customer, but not the individual metering data per time slot. We also implemented the trading protocol in C++ and tested its performance with realistic data. Our simulation results indicate its feasibility for a typical electricity trading period of, for example, 30 minutes, as the bidding protocol is performed for 2500 bids in less than five minutes in the online phase. Both of our protocols are efficient and practical enough to be deployed in real world.

Future work will include possible optimisation of the underlying MPC implementation of the trading protocol.

\bibliographystyle{IEEEtran}
\bibliography{IEEEabrv,TSGbibliography}{}


\end{document}